\documentclass[runningheads]{llncs}
\usepackage[T1]{fontenc}
\usepackage{graphicx}
\usepackage{amsmath,amssymb,amsfonts}
\usepackage{algorithmic}
\usepackage{algorithm}

\usepackage{cite}
\begin{document}
\title{Profit Maximization in Social Networks and Non-monotone DR-submodular Maximization}
%
%
\author{Shuyang Gu\inst{1}\orcidID{0000-0003-4535-2280} \and
Chuangen Gao\inst{2}\and
Jun Huang\inst{3}\and
Weili Wu\inst{4}}
\authorrunning{S. Gu et al.}
%
\institute{Department of Computer Information Systems, Texas A\&M University - Central Texas, Killeen TX 76549, USA \\
	\email{s.gu@tamuct.edu}\\
	\and
School of Computer Science and Technology, Qilu Technology University, China\\
\email{gaochuangen@gmail.com}\\
\and
Department of Computer Science, Baylor University, Waco, TX \\
\email{huangj@ieee.org}
\and
Department of Computer Science, The University of Texas at Dallas, Dallas, TX \\
\email{weiliwu@utdallas.edu}}
\maketitle              
\begin{abstract}
In this paper, we study the non-monotone DR-submodular function maximization over integer lattice. Functions
over integer lattice have been defined submodular property that is similar to submodularity of set functions. 
DR-submodular is a further extended submodular concept for functions
over the integer lattice, which captures the diminishing return property. Such functions find many applications in machine learning, social networks, wireless networks, etc. The techniques for submodular set function maximization can be applied to DR-submodular function maximization, e.g., the double greedy algorithm has a $1/2$-approximation ratio, whose running time is $O(nB)$, where $n$ is the size of the ground set, $B$ is the integer bound of a coordinate.
In our study, we design a $1/2$-approximate binary search double greedy algorithm, and we prove that its time complexity is $O(n\log B)$, which significantly improves the running time. 
Specifically, we consider its application  to the profit maximization problem in social networks with a bipartite model, the goal of this problem is to maximize the net profit gained from a product promoting activity, which is the difference of the influence gain and the promoting cost. We prove that the objective function is DR-submodular over integer lattice.  We apply binary search double greedy algorithm  to this problem and verify the effectiveness.

\end{abstract}
\section{Introduction}

A lot of real-world problems have objective functions with a so-called submodular property, which reflects the diminish return nature of the problems. Since such property exists in a vast amount of applications, submodular optimization has caught a lot of attention during the past two decades. 

Submodular set function optimization includes maximizing or minimizing a submodular function with or without some constraints. One of the directions is non-monotone submodular maximization without constraint. This problem is, given a non-negative submodular set function $f$, to find a subset $S$ that maximizes $f(S)$. Since this problem captures many applications in machine learning, viral marketing, etc., it has been studied extensively. A deterministic local search gives a $1/3$ -approximation and a randomized smoothed local search algorithm gives $2/5$-approximation \cite{feige2011maximizing}. Buchbinder et al. further improve that result, they show that a deterministic double greedy algorithm provides 1/3-approximation, and the randomized version of it gives a 1/2-approximation, both in linear time \cite{buchbinder2015tight}.

Recently, submodular optimization has been extended to functions over integer lattice, which considers the situation that each element in the ground set can be selected as multiple copies. The functions over integer lattice may also have submodular property, which is defined similarly to set functions' submodular. 
Due to the relationship between the two types of submodular, the techniques for submodular set function optimization can be applied to lattice submodular optimization. Based on the  double greedy algorithm in \cite{buchbinder2015tight}, an algorithm for submodular functions over the bounded integer lattice can be designed with $1/3$ approximation ratio \cite{gottschalk2015submodular}. 

An interesting fact is that although the definition of lattice submodular is similar to set function submodular, the lattice submodular does not imply diminish return property like submodular set function. Soma et al. \cite{soma2015generalization} thus give a stronger generalization of submodularity on integer lattice, which is called diminishing return submodular (DR-submodular) functions, such functions capture various applications with diminishing return property. 
In this paper, we specifically study the non-monotone DR-submodular maximization problem and one of the applications: profit maximization in social networks.

The development of online social networks such as Facebook and Twitter provides opportunities for large-scale online viral marketing. Under this circumstance, influence maximization becomes a very popular research direction, which could be described as the problem of finding a small set of most influential nodes in a social network so that the spread of influence in the network is maximizeds. A large amount of efforts have been made on this research topic since Kempe \textit{et al.} \cite{kempe2003maximizing} first defined the problem and obtained plentiful results in many ways\cite{chen2009efficient, gu2020general, gao2020interaction}. 
Those studies are based on the assumption that the number of influenced users determines the profit of product \cite{domingos2001mining, kempe2003maximizing, gao2022adaptive, gao2019robust, ni2021continuous}. 
The influence maximization (IM) problem and the problems extended from IM adopt the cadinality or knapsack constraint, which represents the seeding budget in a product promotion activity. Under such setting, the solution, which is a set of the most influential seed nodes, should maximizes the expected influenced nodes with the constraint of seeds number. Since the influence spread is a non-decreasing function of the seeds set\cite{kempe2003maximizing}, the solution seeds set is always of the limit size. In this paper, we study the overall profit maximization, our goal is to obtain the optimal investment to achieve the maximum net profit, which take into account the varying cost of seeding that is proportional to the number of seeds. In this problem, the objective function is not monotone and there is no cardinality or knapsack constraint.

%
%
%

The contributions of this paper are summarized as follows.
\begin{itemize}
	\item We propose a novel problem named profit  maximization
	\item We prove the profit maximization problem is NP-hard and DR-submodular.
	\item To solve the non-monotone DR-submodular maximization problem, we propose the binary search double greedy algorithm.
	\item We prove the algorithm has a 1/2- approximate ratio and the time complexity is polynomial($n\log B$). To the best of our knowledge, this is the fastest algorithm with the least queries to the objective function.
\end{itemize}

This article is organized as follows. In Section 2, we review the existing work. In section 3, we present the preliminaries for submodular and DR submodular optimization problems. Then we introduce the profit maximization problem in section 4 and show the objective function is  DR-submodular. In section 5, the binary search double greedy algorithm is presented to solve the DR-submodular maximization problem, we give theoretical proof of the approximation ratio and time complexity. The conclusion is presented in Section 6.

\section{Related Work}
 Non-monotone DR-submodular Maximization is closely related to non-monotone submodular set function maximization because the algorithm for the latter problem can be applied to the former problem directly. The non-monotone submodular maximization is also called Unconstrained Submodular Maximization(USM). USM has various applications, such as marketing strategies  over social networks \cite{hartline2008optimal}, Max-Cut\cite{goemans1995improved}, and maximum facility location \cite{ageev19990}. USM problem has been studied extensively\cite{ gharan2011submodular, buchbinder2018deterministic, pan2014parallel}.  Buchbinder gives a tight linear time randomized (1/2)-approximation for the problem\cite{buchbinder2015tight}.
 
  The topic of functions over integer lattice optimization has attracted much attention recently, the submodularity of such functions are considered. Monotone submodular functions over integer lattice with cardinality constraint are addressed in \cite{soma2018maximizing, lai2019monotone, zhang2021streaming}. Sahin et al. study lattice submodular functions subject to a discrete (integer) polymatroid constraint \cite{sahin2020constrained}. Zhang et al. study the problem of maximizing the sum of a monotone non-negative DR-submodular function and a supermodular function on the integer lattice subject to a cardinality constraint\cite{zhang2021maximizing}. The non-submodular functions on the integer lattice are addressed in \cite{kuhnle2018fast}. Nong et al. focus on maximizing a non-monotone weak-submodular function on a bounded integer lattice \cite{nong20201}. For the problem addressed in this paper, non-monotone DR-submodular function maximization, Soma et al. design a $\frac{1}{2+\epsilon}$-approximation algorithm with 
  a running time of $O(\frac{n}{\epsilon} \log^2B)$ \cite{soma2017non}.

In the meantime, the discrete domains of submodular functions over integer lattice are further extended to continuous domains, Hassani et al. study
stochastic projected gradient methods for maximizing continuous submodular functions with convex constraints \cite{hassani2017gradient}. In \cite{bian2017continuous, niazadeh2018optimal}, the authors consider maximizing a continuous and nonnegative submodular function over a hypercube.

 For optimization problems in social networks, many fall into submodular set function maximization subject to a cardinality constraint \cite{kempe2003maximizing, gu2020general, gao2022adaptive, gao2019robust}. The budget allocation problem \cite{soma2014optimal}  falls into monotone submodular function maximization over integer lattice subject to a knapsack constraint.

\section{Preliminaries}

We say that a set function $g:  2^E \to \mathbb{R_+} $ is submodular if it satisfies following inequality: 
\begin{align*}
g(X) + g(Y) \geq g (X \cup Y ) + g (X \cap Y),  \forall X, Y \subseteq E
 \end{align*} 

The submodular definition also reflects a natural “diminishing returns” property: the marginal gain of adding an element to a set $X$ is at least as high as the marginal gain of adding the same element to a superset of $X$. Formally, 
for every set $X,Y $ such that $X \subseteq Y \subseteq E$ and every $e \in E \setminus Y$, it follows that
$$g(X\cup\{e\})-g(X) \ge g(Y\cup \{e\})-g(Y)$$ 

The monotonicity definition for a set function is that a function is monotone non-decreasing iff $ g(X) \le g(Y)$ for $\forall X \subseteq Y$. 

Functions over integer lattice may have similar property as submodular set functions. A function  $h:\mathbb{Z}_+^E \to \mathbb{R_+} $ that is defined over the integer lattice is submodular if the following holds\cite{gottschalk2015submodular}:
\begin{align*}
h(\boldsymbol{x} ) + h(\boldsymbol{y} ) \geq h(\boldsymbol{x}  \lor\boldsymbol{y} ) + h (\boldsymbol{x}  \land \boldsymbol{y} ),   \boldsymbol{x} ,\boldsymbol{y}   \in \mathbb{Z}_+^E.
\end{align*} 
where $(\boldsymbol{x} \lor \boldsymbol{y})(i) = \max\{\boldsymbol{x} (i),\boldsymbol{y} (i)\}$ and $(\boldsymbol{x} \land \boldsymbol{y} )(i) = \min\{\boldsymbol{x} (i), \boldsymbol{y} (i)\}$. The vectors $\boldsymbol{x} $ and $\boldsymbol{y}$ have the same dimensions.  Hence $\boldsymbol{x} \lor \boldsymbol{y}$ represents coordinate-wise maximum, and  $\boldsymbol{x} \land \boldsymbol{y}$ denote the
coordinate-wise minimum. We can see this form of submodularity is a more generalized definition of submodularity that covers set functions submodular, because vectors with all entries equal to either 0 or 1 can be seen as a subset including the elements that are equal to $1$ while excluding the elements that are equal to $0$, in that case, $\boldsymbol{x} \land \boldsymbol{y}$ and $\boldsymbol{x} \lor \boldsymbol{y}$ transform to set intersection and set union of the subsets that $\boldsymbol{x}$ and $\boldsymbol{y}$ represent respectively . 

The submodular function over integer lattice does not have the diminishing return property. To capture such property in real-world problems, a stronger version of submodularity has been introduced, which is called DR-submodular\cite{soma2015generalization}. DR submodular function on a bounded integer lattice satisfies the following diminish return property: 
\begin{align*}
h(\boldsymbol{x} +\chi_e)-h(\boldsymbol{x} )\geq h(\boldsymbol{y} +\chi_e)-h(\boldsymbol{y} ), \forall \boldsymbol{x} \leq \boldsymbol{y},  \forall e \in E
\end{align*} 
where $\chi_e$ denotes a unit vector, i.e. $\chi_e \in \mathbb{Z}^E$ is the vector with $\chi_e(e) = 1 $ and $\chi_e(a) = 0$ for every $a \neq e$.

The problem we consider is maximizing (non-monotone) DR-submodular functions over bounded integer lattice. Formally, we study the optimization problem
	\begin{eqnarray}
\begin{aligned}
\max && f(x) \\
\mbox{subject to} && \boldsymbol{0}\leq x \leq \boldsymbol{B},
\end{aligned}
	\end{eqnarray}
where $f : \mathbb{Z}_+^E \to \mathbb{R}_+$ is a non-negative DR-submodular
function and not necessarily monotone.  $\boldsymbol{0}$ is the all zero vector, and $\boldsymbol{B} \in \mathbb{Z}_+^E$ is a vector representing the maximum value for each coordinate. When $\boldsymbol{B}$ is the all-ones vector, the problem is equivalent to the original unconstrained submodular set function maximization. We assume that $f$ is given as an evaluation oracle; when we specify $\boldsymbol{x} \in \mathbb{Z}_+^E$, the oracle returns the value of $f(\boldsymbol{x})$. We define  $f( \boldsymbol{x}|\boldsymbol{y}) = f( \boldsymbol{y} + \boldsymbol{x}) - f( \boldsymbol{y})$.

\section{Profit  Maximization Problem}\label{formulation}
In this section, we formulate  profit  maximization problem in bipartite model formally and prove  the objective function  is DR-submodular.

We adopted the bipartite influence model \cite{alon2012optimizing}, by which a social network can be modeled as a bipartite graph. There are two types of nodes, the source nodes (marketing channels) and the target nodes (potential customers). The marketing channels include the social medias, TV, newspapers etc., which may initiat marketing influence to potential customers. Each marketing channels can affect a subset of  potential customers, which are the customers follows the channel. For example,  a social media account has a group of followers, a TV channel has relatively fixed audience. 

Another assumption is that each marketing channel may influence its potential customers multiple times, this assumption makes sense since a social media account may post a product promotion advertisement multiple times, so does a TV channel, a newspaper or other marketing channels. Each time when a marketing channel initiate an influence, the influence propogate to all the audience of that channel. If a marketing channel initiate $x$ trials of influence, then a viewer recieves the advertisement $x$ times, we assume the $x$ trials of influence events are independent .

The bipartite graph is denoted $G = (S,T;E)$, where $S$ is the set of source nodes(marketing channels), $T$ is the set of target nodes(potential customers), $E\subseteq S \times T$ is the edge set. An
edge of the bipartite graph between a source node  $i$ and a target node  $j$  indicates $i$ may influence $j$ with some probability. If a customer does not follow some channel, there is no edge between them. Each source node s has a capacity $c(s) \in  \mathbb{N_+}$, which stands for the trials limit at that source node.
For each source node $s$, there are a series of probabilities $p_s^1$, $p_s^2$, $ \dots p_s^{c(s)}$, where $p_s^{(i)}$ represents the probability to activate any adjacent target node of the $s$'s $i$th  trials of influence. 

\begin{figure}[tbp]	
	\centerline{\includegraphics[scale=0.4]{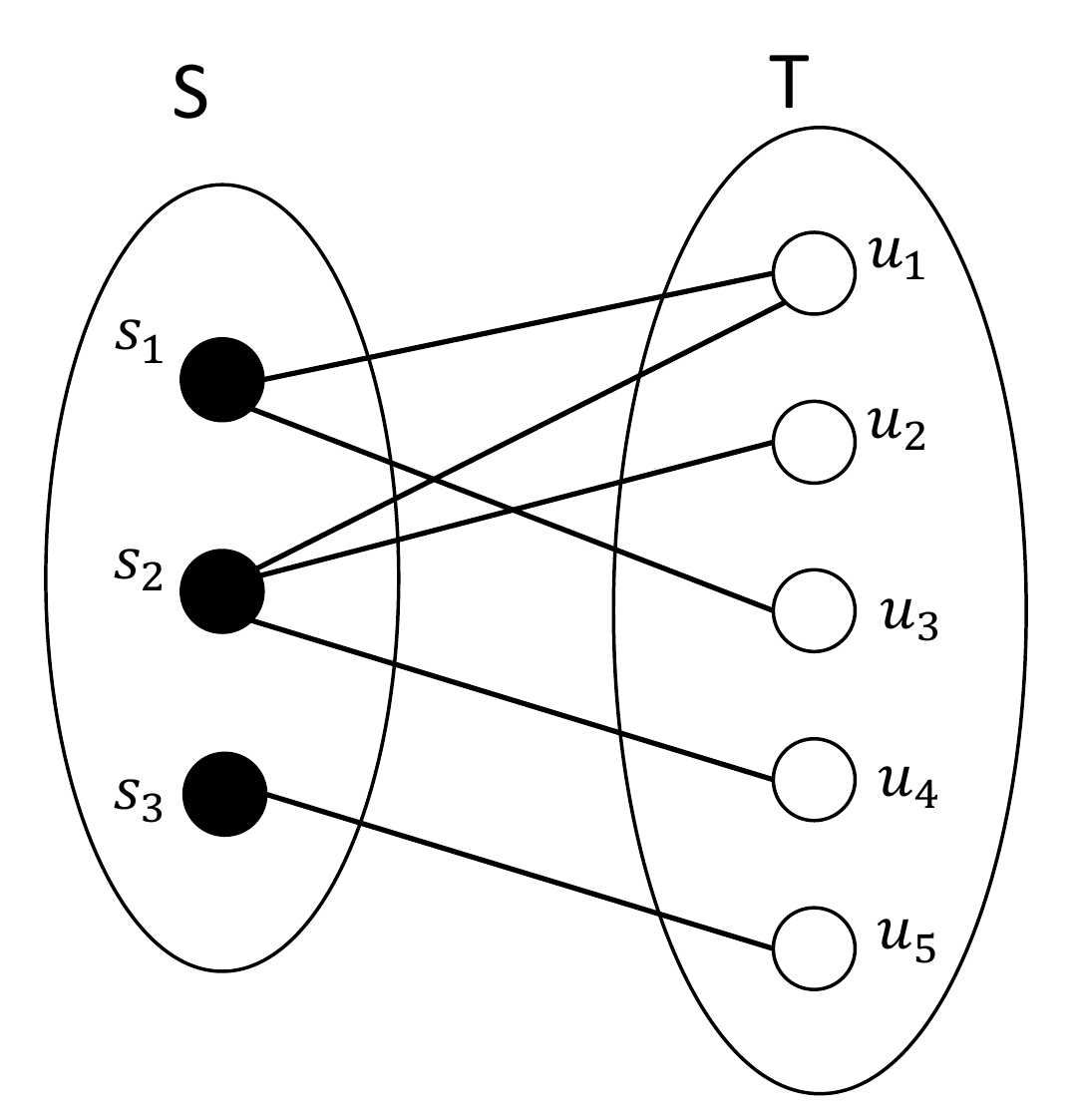}}
	\caption{A bipartite social network example}
	\label{Fig.1}
\end{figure}

A mini social network shown in figure \ref{Fig.1} is built as an example, there are three marketing channels and five potential customers, customer $u_1$ and $u_3$ are the followers of the marketing channel $s_1$,  customers $u_1$,  $u_2$  and $u_4$ are the followers of marketing channel $s_2$,  customer $u_5$ is a follower of marketing channel $s_3$. Suppose channel $s_1$ has a marketing capacity $c(s_1) =10$, which means the marketing influence trials can be sent are at most 10 times from $s_1$. Suppose $c(s_2) =20$, $c(s_3) =15$. For node $s_1$, $p_{s_1}^1$, $p_{s_1}^2$, $ \dots p_{s_1}^{10}$ are given. Similarly, we are given $p_{s_2}^1$, $p_{s_2}^2$, $ \dots p_{s_1}^{20}$ and $p_{s_3}^1$, $p_{s_3}^2$, $ \dots p_{s_3}^{15}$.

The goal is to find the optimal number of influence trials for each source node respecting the capacity of that node, and maximizes the overall profit. We call the solution a marketing strategy, which is denoted as the vector $\textbf{m}$,  thus $\textbf{m}(s)$ is the number of influencing trials at the source node $s$. Note that for one potential customer, he/she may be connected to more than one marketing channels. 

In the above example, if a marketing strategy $\textbf{m}= (2, 1, 0)$ is applied, node $u_1$ receives the influence three times, two times from the source node $s_1$, one time from $s_2$, the three events are independent.
The probability that the node $u_1$ is activated is $1-(1-p_{s_1}^1)\cdot  (1-p_{s_1}^2) \cdot (1-p_{s_2}^1)$. 

For the general problem, let $\Gamma(t)$ be the set of source nodes that are adjacent to node $t$. The probability that the node $t$ is activated is:
\begin{align}
p_t({\textbf{m}}) =1- \prod_{s \in \Gamma(t)} \prod_{i=1}^{\textbf{m}(s)} {(1-p^i_s)}
\end{align}

The optimization target is the expected overall profit, which should be the difference between the expected influenced nodes $\sigma({\textbf{m}})$ and the cost of marketing $\delta({\textbf{m}})$. We can obtain $\sigma({\textbf{m}})$ by summing $p_t({\textbf{m}}) $ for all $t\in T$,

\begin{align}\label{eq:2}
\sigma({\textbf{m}}) = \sum_{t \in T} p_t({\textbf{m}}) 
=\sum_{t \in T} {(1- \prod_{s \in \Gamma(t)} \prod_{i=1}^{\textbf{m}(s)} {(1-p^i_s)})}
\end{align}

The marketing cost $\delta({\textbf{m}})$ is proportional to the number of influencing trials, let $\delta_s$ be the unit cost for each trial at the marketing channel $s$.  Then the total cost of all marketing channels is

\begin{align}\label{eq:3}
\delta({\textbf{m}}) = \sum_{s \in S} \textbf{m}(s) \cdot \delta_s
\end{align}

Hence, the objective function is

\begin{align}\label{eq:4}
f({\textbf{m}}) = \sigma({\textbf{m}}) - \delta({\textbf{m}}) =
\sum_{t \in T} {(1- \prod_{s \in \Gamma(t)} \prod_{i=1}^{\textbf{m}(s)} {(1-p^i_s)})} -  \sum_{s \in S} \textbf{m}(s) \cdot \delta_s
\end{align}

\begin{definition}[Profit Maximization Problem(PM)]
	Given a social network $G=(S,T;E)$,  the marketing capacity vector $\textbf{c}$,  the influencing probabilities $p_s^1, p_s^2,\dots,p_s^{c(s)}$ for each node $s\in S$, and the unit marketing cost $\delta_s$,  find a marketing strategy $\textbf{m}$, which represents the number of influencing trials at the marketing channels, to maximize the expected profit through influence propagation:
	\begin{eqnarray}
	\max f(\textbf{m})\\
	s.t. \textbf{0}\leq\textbf{m}\leq \textbf{c}
	\end{eqnarray}
	
\end{definition}

Next we show the objective function of profit maximization problem is DR-submodular. We assume $p_e^i \leq p_e^j$, for any $e \in E$, and $i\geq j$, because an advertisement has less marketing effect at the $i$th time people watch it than at an earlier $j$th time.  

\begin{theorem}
	$f(\textbf{m})$ is non-monotone DR-submodular.
\end{theorem}

The proof of Theorem 1 relies upon two lemmas.  

\begin{lemma} \label{lemma1}
	$\sigma({\textbf{m}})$ is DR-submodular.
\end{lemma}

{\em Proof.}\quad  It suffices if $p_t({\textbf{m}})$ is DR-submodular, because $\sigma({\textbf{m}})$ is a linear summation of $p_t({\textbf{m}})$, $\forall t \in T$. Suppose ${\textbf{x}}\leq {\textbf{y}}$, ${\textbf{x}}, {\textbf{y}} \in\mathbb{Z}_+^S $. We have
\begin{align*}
p_t({\textbf{x}})= 1- \prod_{s \in \Gamma(t)} \prod_{i=1}^{\textbf{x}(s)} {(1-p_s^i)}
\end{align*}
\begin{align*}
p_t({\textbf{x}}+\chi_e)= 1- [\prod_{s \in \Gamma(t)} \prod_{i=1}^{\textbf{x}(s)} {(1-p^i_s)}]\cdot (1-p_e^{\textbf{x}(s)+1})
\end{align*}

Thus,

\begin{align*}
p_t(\chi_e|{\textbf{x}}) = p_t({\textbf{x}}+\chi_e) - p_t({\textbf{x}})=[\prod_{s \in \Gamma(t)} \prod_{i=1}^{\textbf{x}(s)} {(1-p_s^i)}]\cdot p_e^{\textbf{x}(e)+1}
\end{align*}
Similarly, we have
\begin{align*}
p_t(\chi_e|{\textbf{y}}) = p_t({\textbf{y}}+\chi_e) - p_t({\textbf{y}})=[\prod_{s \in \Gamma(t)} \prod_{i=1}^{\textbf{y}(s)} {(1-p_s^i)}]\cdot p_e^{\textbf{y}(e)+1}
\end{align*}

Since  ${\textbf{x}}\leq {\textbf{y}}$, so $\textbf{x}(s) \leq \textbf{y}(s)$, therefore
\begin{align*}
\prod_{s \in \Gamma(t)} \prod_{i=1}^{\textbf{x}(s)} {(1-p_s^i)} \geq \prod_{s \in \Gamma(t)} \prod_{i=1}^{\textbf{y}(s)} {(1-p_s^i)}
\end{align*}
Since  ${\textbf{x}}\leq {\textbf{y}}$, we have $\textbf{x}(e)+1 \leq \textbf{y}(e)+1$. So $p_e^{\textbf{x}(e)+1} \geq p_e^{\textbf{y}(e)+1}$. Hence,
\begin{align*}
p_t({\textbf{x}}+\chi_e) - p_t({\textbf{x}}) \geq  p_t({\textbf{y}}+\chi_e) - p_t({\textbf{y}})
\end{align*}

That means  $ p_t({\textbf{m}})$ is DR-submolar, and so does 	$\sigma({\textbf{m}})$.
\qed

\begin{lemma} \label{lemma2}
	If  function $h(\textbf{x})$ is DR-submodular, function $l(\textbf{x})$ is modular, then the function $f(\textbf{x}) = h(\textbf{x})-l(\textbf{x})$ is DR-submodular.
\end{lemma}
{\em Proof.}\quad  
\begin{align}
f({\boldsymbol{x}}+\chi_e) - f({\boldsymbol{x}}) &= h({\boldsymbol{x}}+\chi_e) - l({\boldsymbol{x}}+\chi_e) -[h({\boldsymbol{x}}) - l({\boldsymbol{x}}) ] \\& = h({\boldsymbol{x}}+\chi_e) -h({\boldsymbol{x}}) - [ l({\boldsymbol{x}}+\chi_e)- l({\boldsymbol{x}})]\\&\geq h({\boldsymbol{y}}+\chi_e) - h({\boldsymbol{y}})-[l({\boldsymbol{x}}+\chi_e) - l({\boldsymbol{x}}) ]\\&=h({\boldsymbol{y}}+\chi_e) - h({\boldsymbol{y}})-[l({\boldsymbol{y}}+\chi_e) - l({\boldsymbol{y}}) ]\\& = f({\boldsymbol{y}}+\chi_e) - f({\boldsymbol{y}})
\end{align}

The inequality (10) holds because of $h$'s DR-submolarity, the equality (11) is due to  $l$'s modularity.
\qed

Based on the above two lemmas, we can easily conclude that the objective function $f({\textbf{m}})$ is DR-submolar since $\sigma({\textbf{m}}) $ is DR-submodular and $\delta({\textbf{m}})$ is modular. It is apparent that both the function  $\sigma({\textbf{m}}) $ and $\delta({\textbf{m}})$ are monotone non-decreasing,  the difference between the monotone submodular functions and the monotone modular function may not be monotone. 

For the computational complexity of the proposed problem,  we know that the NP-hard Max Cut problem is a special case of the unconstrained submodular set function maximization \cite{feige2011maximizing}, therefore the more general form lattice submodular function maximization is also NP-hard.

\section{Algorithm}
In this section, we present the algorithm for non-monotone DR-submodular function maximization. The main idea is inspired by the double greedy algorithm for the unconstrained submodular maximization(USM)\cite{buchbinder2015tight} on set functions. The algorithm can be extended to accommodate DR-submodular function over integer lattice\cite{soma2017non}, because DR-submodular function can be treated as submodular set function on each coordinate. We investigate some interesting properties for DR-submodular functions to further speed up the algorithm.

 We define two functions $\phi(b)  :=f(\boldsymbol{\chi}_e|\boldsymbol{x}+b\boldsymbol{\chi}_e)$ ,$\psi(b)  :=f(-\boldsymbol{\chi}_e|\boldsymbol{y}-b\boldsymbol{\chi}_e)$, where $b \in \mathbb{Z}^+ $. Both functions are non-increasing functions of $b$ because the function $f$ is DR-submodular.

\begin{algorithm}[htbp]
	\caption{Binary Search Greedy Algorithm}
	\textbf{Input: $f : \mathbb{Z}_+^E \to \mathbb{R}^+$, $\boldsymbol{B} \in \mathbb{Z}_+^E$} \\
	\textbf{Assumption:} \text{ $f$ is DR-submodular} \\
	
	\begin{algorithmic}[1]
		\STATE $\boldsymbol{x}\leftarrow \boldsymbol{0}$, $\boldsymbol{y}\leftarrow \boldsymbol{B}$;
		\FOR {$e\in E$}
		\STATE Find $\arg \min_b{\phi(b)}$ such that  $\phi(b)< 0 $ by binary search.  
        \STATE$u\leftarrow \boldsymbol{x}(e)+\arg \min_b\phi(b)-1$.
		\STATE Find $\arg \min_b{\psi(b)}$ such that  $\psi(b)< 0 $ by binary search. 
		\STATE$v\leftarrow \boldsymbol{y}(e)-\arg \min_b\psi(b)+1$.
		\WHILE{$\boldsymbol{x}(e) < \boldsymbol{y}(e)$}
 	    \STATE $\sigma \leftarrow \max(\lfloor \frac{\boldsymbol{y}(e)-\boldsymbol{x}(e)}{2}\rfloor,1) $		
 	    \STATE $\alpha \leftarrow f(\sigma \boldsymbol{\chi}_e |\boldsymbol{x})$ and $\beta \leftarrow f(-\sigma  \boldsymbol{\chi}_e  |\boldsymbol{y})$
		\IF{$\beta\leq 0$}
		\STATE $\boldsymbol{x}(e) \leftarrow \boldsymbol{x}(e) + \sigma$
		\ELSIF{$\alpha \leq 0$}
		\STATE $\boldsymbol{y}(e) \leftarrow \boldsymbol{y}(e) - \sigma$
		\ELSE
		\STATE  Randomly update $\boldsymbol{x}(e)\leftarrow \boldsymbol{x}(e)+\sigma$ or $\boldsymbol{y}(e)\leftarrow \boldsymbol{y}(e)-\sigma$; the former case occurs with probability $\frac{\alpha}{\alpha+\beta}$, the later case with the probability $\frac{\beta}{\alpha+\beta}$. 
		\ENDIF
		\ENDWHILE
		\IF{$\boldsymbol{x}(e)\geq u$}
		\STATE  $\boldsymbol{x}(e)\leftarrow u$, $\boldsymbol{y}(e)\leftarrow u$
			\ENDIF
		\IF{$\boldsymbol{y}(e)\leq v$}
		\STATE  $\boldsymbol{y}(e)\leftarrow v$, $\boldsymbol{x}(e)\leftarrow v$
			\ENDIF
		\ENDFOR
		\RETURN $\boldsymbol{x}$
	\end{algorithmic}
\end{algorithm}

The algorithm starts with two vectors,  $x = \boldsymbol{0}$ and $y = \boldsymbol{c}$. For each coordinate $e \in E$ it iteratively either increase ${\textbf{x}}(e)$ or decrease ${\textbf{y}}(e)$ by $\sigma$, which depends on the marginal gain by adding $\sigma$ to ${\textbf{x}}(e)$ and the marginal gain by removing $\sigma$ to ${\textbf{y}}(e)$. This procedure continues until ${\textbf{x}}(e) = {\textbf{y}}(e)$. Then it moves on to work on the next coordinate, after ${\textbf{x}}$ and ${\textbf{y}}$ agrees on all coordinates $e\in E$, ${\textbf{x}}={\textbf{y}}$, and the vector is the output of the algorithm. Different from applying the double greedy algorithm directly, which tightens the gap one unit per step, algorithm 1 tightens it by half in each iteration. The binary search nature of this algorithm guarantees that the number of iterations needed is a logarithm of $B$.  Next let us firstly give a few results based on the diminish return property, which will be used in proving the theoretical guarantee of the algorithm later.


\begin{lemma} \label{lemma1}
For $\forall \boldsymbol{x}\leq \boldsymbol{y}, k\geq 1, k\in \mathbb{Z}^+$, and the value of $k$ does not violate the integer bound. We have
\begin{align*}
f(\boldsymbol{x}+k\boldsymbol{\chi}_e)-f(\boldsymbol{x})\geq f(\boldsymbol{y}+k\boldsymbol{\chi}_e)-f(\boldsymbol{y}) 
\end{align*}
\end{lemma} 

\begin{proof}
	By the definition of DR-submodular, we have
	\begin{eqnarray}
\begin{aligned}
f(\boldsymbol{x}+\boldsymbol{\chi}_e)-f(\boldsymbol{x})&\geq f(\boldsymbol{y}+\boldsymbol{\chi}_e)-f(\boldsymbol{y}) \\
f(\boldsymbol{x}+2\boldsymbol{\chi}_e)-f(\boldsymbol{x}+\boldsymbol{\chi}_e)&\geq f(\boldsymbol{y}+2\boldsymbol{\chi}_e)-f(\boldsymbol{y}+\boldsymbol{\chi}_e) \\
&\vdots\\
f(\boldsymbol{x}+k\boldsymbol{\chi}_e)-f(\boldsymbol{x} + (k-1)\boldsymbol{\chi}_e)&\geq f(\boldsymbol{y}+k\boldsymbol{\chi}_e)-f(\boldsymbol{y}+(k-1)\boldsymbol{\chi}_e) \nonumber
\end{aligned}
\end{eqnarray}
	Sum up the above inequalities, the lemma holds. \qed
\end{proof}

Note that based upon lemma \ref{lemma1}, we have
\begin{eqnarray}
\begin{aligned}
\alpha + \beta = f(\boldsymbol{x}+\sigma\boldsymbol{\chi}_e)-f(\boldsymbol{x})-(f(\boldsymbol{y}+\sigma\boldsymbol{\chi}_e)-f(\boldsymbol{y}))\geq 0
\end{aligned}
\end{eqnarray}

\begin{lemma} \label{lemma2}
	For $\forall \boldsymbol{x}\leq \boldsymbol{y}, k\geq 1, k\in \mathbb{Z}^+$, $k\leq \boldsymbol{x}(e)$, we have
	\begin{align*}
	f(\boldsymbol{x}-k\boldsymbol{\chi}_e)-f(\boldsymbol{x})\leq f(\boldsymbol{y}-k\boldsymbol{\chi}_e)-f(\boldsymbol{y}) 
	\end{align*}
\end{lemma} 
 Lemma \ref{lemma2} can be easily obtained from lemma 
 \ref{lemma1}.

%
%
%
%

\begin{lemma}\label{lemma3}
Given $p\leq q$, $p, q \in Z^+$, if $f(\boldsymbol{\chi}_e|\boldsymbol{x}+(q-1)\boldsymbol{\chi}_e)\geq 0$, then
	\begin{align*}
	0\leq f(p\boldsymbol{\chi}_e|\boldsymbol{x})\leq f(q\boldsymbol{\chi}_e|\boldsymbol{x})
	\end{align*}
\end{lemma}
\begin{proof}
	\begin{eqnarray}
	\begin{aligned}
	f(p\boldsymbol{\chi}_e|\boldsymbol{x}) = f(\boldsymbol{\chi}_e|\boldsymbol{x})+f(\boldsymbol{\chi}_e|\boldsymbol{x}+\boldsymbol{\chi}_e)+\dots + f(\boldsymbol{\chi}_e|\boldsymbol{x}+(p-1)\boldsymbol{\chi}_e)\\
		f(q\boldsymbol{\chi}_e|\boldsymbol{x}) = f(\boldsymbol{\chi}_e|\boldsymbol{x})+f(\boldsymbol{\chi}_e|\boldsymbol{x}+\boldsymbol{\chi}_e)+\dots + f(\boldsymbol{\chi}_e|\boldsymbol{x}+(q-1)\boldsymbol{\chi}_e)\nonumber
	\end{aligned}
	\end{eqnarray}

	Since the values of the terms in the above equations are non-increasing from left to right, $f(\boldsymbol{\chi}_e|\boldsymbol{x}+(q-1)\boldsymbol{\chi}_e)\geq 0$, so all terms are greater than or equal to 0. And $f(q\boldsymbol{\chi}_e|\boldsymbol{x}) $ has more terms, so the lemma holds.\qed
	\end{proof}

We can obtain a similar property in terms of the vector $\boldsymbol{y}$.
\begin{lemma}\label{lemma4}
	Given $p\leq q$, $p, q \in Z^+$, if $f(-\boldsymbol{\chi}_e|\boldsymbol{y}-(q-1)\boldsymbol{\chi}_e)\geq 0$, then
	\begin{align*}
	0\leq f(-p\boldsymbol{\chi}_e|\boldsymbol{y})\leq f(-q\boldsymbol{\chi}_e|\boldsymbol{y})
	\end{align*}
\end{lemma}


The rest of this section is devoted to proving that Algorithm 1 provides an approximation ratio of $1/2$ for DR-submodular maximization. Let us begin the analysis of Algorithm 1 with the introduction of some notation. Let $\boldsymbol{x}_i^e$ and $\boldsymbol{y}_i^e$ be  random variables denoting the vectors  generated by the algorithm at the end of the i-th iteration for coordinate $e$, let the number of iterations for coordinate $e$ is $\theta_e$, note that $1\leq i\leq \theta_e\leq \log B$. Denote by $\boldsymbol{opt}$ the optimal solution. Let us define the following random variable: $\boldsymbol{opt}_i^e \triangleq(\boldsymbol{opt}\lor \boldsymbol{x}_i^e )\land \boldsymbol{y}_i^e$. Note that $\boldsymbol{x}_0^{e}(e)=0$,  $\boldsymbol{y}_{0}^{e}(e)=B$, and $\boldsymbol{opt}_0^{e} (e)=\boldsymbol{opt}(e)$. Additionally, the following always holds: $\boldsymbol{opt}_{\theta_e}^e(e) = \boldsymbol{x}_{\theta_e}^e(e)  = \boldsymbol{y}_{\theta_e}^e(e) $, $\forall e \in E$. 

Let us analyze the approximation ratio of the randomized algorithm. We consider the subsequence $\mathbb{E}[f(\boldsymbol{opt}_0^e)],\dots,\mathbb{E}[f(\boldsymbol{opt}_{\theta_e}^e)]$ for any dimension $e \in E$, and a whole sequence which is a combination of every such subsequence for each element $e \in E$. This sequence starts with $f(\boldsymbol{opt})$ and ends with the expected value of the algorithm’s output. The following lemma upper bounds the loss between every two consecutive elements in the sequence. Formally, $\mathbb{E}[f(\boldsymbol{opt}_{i-1}^e)-f(\boldsymbol{opt}_{i}^e)]$ is upper bounded by the average expected change in the value of the two solutions maintained by the algorithm, i.e., $\frac{1}{2} \mathbb{E} [f(\boldsymbol{x}_i^e) -f(\boldsymbol{x}_{i-1}^e)  +f(\boldsymbol{y}_i^e) -f(\boldsymbol{y}_{i-1}^e)]$. 

\begin{lemma}\label{lemma5}
	For every $1\leq i \leq \theta_e$,
	\begin{eqnarray}\label{eq1}
	\mathbb{E}[f(\boldsymbol{opt}_{i-1}^e)-f(\boldsymbol{opt}_{i}^e)]\leq \frac{1}{2} \mathbb{E} [f(\boldsymbol{x}_i^e) -f(\boldsymbol{x}_{i-1}^e)  +f(\boldsymbol{y}_i^e) -f(\boldsymbol{y}_{i-1}^e)]
	\end{eqnarray}
	where expectations are taken over the random choices of the algorithm.
\end{lemma}
\begin{proof}
	Notice that it suffices to prove the inequality conditioned on any event of the form $\boldsymbol{x}_{i-1}^e=\boldsymbol{s}_{i-1}^e$, where $\boldsymbol{s}_{i-1}^e \in \mathbb{Z_+^E}$, for which the probability that $\boldsymbol{x}_{i-1}^e = \boldsymbol{s}_{i-1}^e$ is nonzero. Hence, fix such an event corresponding to an integer vector $\boldsymbol{s}_{i-1}^e$. The rest of the proof implicitly assumes everything is conditioned on this event. Since the analysis is same for every coordinate, we omit the superscript $e$ in $x_i^e, y_i^e, s_i^e$ and $ opt_i^e$ in the following proof. After an iteration $i$ on coordinate $e$, denote by $\delta_i$ the distance between $x_i(e)$ and $y_i(e)$, which can be calculate as $\delta_i=y_i(e)-x_i(e)=B-\sum_{ k=1}^i\sigma_{k}$. The parameter $\sigma_i$ can be obtained iteratively as $\sigma_1=\lfloor \frac{B}{2}\rfloor$, $\sigma_i=\lfloor \frac{B-\sum_{k=1}^{i-1}\sigma_k}{2}\rfloor=\lfloor\frac{\delta_i}{2}\rfloor$. Due to the conditioning, the following random variables become constants:
	
	\begin{enumerate}
		\item	 $\boldsymbol{y}_{i-1}$,  where $\boldsymbol{y}_{i-1}(e) = \boldsymbol{s}_{i-1}(e)+\delta_{i-1}$
		\item	 $\boldsymbol{opt}_{i-1}\triangleq (\boldsymbol{opt} \lor \boldsymbol{x}_{i-1}) \land \boldsymbol{y}_{i-1}$, where $\boldsymbol{opt}_{i-1}(e)= \boldsymbol{s}_{i-1}(e)+\min (\boldsymbol{opt}(e)-\boldsymbol{s}_{i-1}(e), \delta_{i-1})$
		\item	 $\alpha_i$ and $\beta_i$, which refer to $\alpha, \beta$ at the iteration $i$ .
	\end{enumerate}

	By Lemma \ref{lemma1}, $\alpha_i+\beta_i\geq 0$. Thus at most one of $\alpha_i$ and $\beta_i$ is strictly less than zero. We need to consider the following three cases for the value of $\alpha_i$ and $\beta_i$:
	
	\textbf{Case 1}: ($\alpha_i \geq 0 $ and $\beta_i\leq0$). In this case  the vector $\boldsymbol{y}$ does not change: $\boldsymbol{y}_i= \boldsymbol{y}_{i-1}$. The vector $\boldsymbol{x}$ changes. $\boldsymbol{x}_i \leftarrow \boldsymbol{x}_{i-1}+ \sigma_i\boldsymbol{\chi}_e$. Hence, $f(\boldsymbol{y}_i)-f(\boldsymbol{y}_{i-1}) = 0$. Also, by our definition $\boldsymbol{opt}_{i}\triangleq (\boldsymbol{opt} \lor \boldsymbol{x}_{i}) \land \boldsymbol{y}_{i}=(\boldsymbol{opt} \lor (\boldsymbol{x}_{i-1}+\sigma_i  \boldsymbol{\chi}_e)) \land \boldsymbol{y}_{i}$. Thus, we are left to prove that
	\begin{eqnarray}
	f((\boldsymbol{opt} \lor \boldsymbol{x}_{i-1}) \land \boldsymbol{y}_{i-1})-f((\boldsymbol{opt} \lor (\boldsymbol{x}_{i-1}+\sigma_i  \chi_e)) \land \boldsymbol{y}_{i})\nonumber\\\leq  \frac{1}{2} [f(\boldsymbol{x}_i) -f(\boldsymbol{x}_{i-1})]= \frac{\alpha_i}{2}
	\end{eqnarray}
	We prove it  considering the relationship among $\boldsymbol{x}_{i-1}(e)$, $\boldsymbol{x}_i(e)$ and $\boldsymbol{opt}_i(e)$.
	
	\textbf{Case 1.1}:
	$\boldsymbol{x}_i(e)=\boldsymbol{s}_{i-1}
	(e)+\sigma_i \leq \boldsymbol{opt}(e)$.
	
	This condition implies $\boldsymbol{x}_{i-1}
	(e)\leq \boldsymbol{opt}(e)$. Since $\boldsymbol{y}_i= \boldsymbol{y}_{i-1}$, the left-hand side of the inequality (15) is 0, which is definitely not greater than the nonnegative $\frac{\alpha_i}{2}$.
	
	\textbf{Case 1.2}: $\boldsymbol{s}_{i-1}(e)\geq  \boldsymbol{opt}(e)$.
	
	This condition implies that  $\boldsymbol{x}_i(e)=\boldsymbol{s}_{i-1}(e)+\sigma_i > \boldsymbol{opt}(e)$. 

We can see that $(\boldsymbol{opt} \lor (\boldsymbol{x}_{i-1}+\sigma_i \boldsymbol{\chi}_e)) \land \boldsymbol{y}_{i}=(\boldsymbol{opt} \lor \boldsymbol{x}_{i-1}) \land \boldsymbol{y}_{i-1}+\sigma_i  \boldsymbol{\chi}_e=\boldsymbol{opt}_{i-1}+\sigma_i \boldsymbol{\chi}_e\leq  \boldsymbol{y}_{i-1}$, by diminish return submodularity we have
	\begin{eqnarray}
	\begin{aligned}
	f((\boldsymbol{opt} \lor \boldsymbol{x}_{i-1}) \land \boldsymbol{y}_{i-1})-f((\boldsymbol{opt} \lor (\boldsymbol{x}_{i-1}+\sigma_i \boldsymbol{\chi}_e)) \land \boldsymbol{y}_{i})\\=f(-\sigma_i  \boldsymbol{\chi}_e|\boldsymbol{opt}_{i-1}+\sigma_i \boldsymbol{\chi}_e)\leq f(-\sigma_i \boldsymbol{\chi}_e|\boldsymbol{y}_{i-1})=\beta\leq 0\leq \frac{\alpha_i}{2} \nonumber
	\end{aligned}
	\end{eqnarray}
	
	\textbf{Case 1.3}:
	$\boldsymbol{s}_{i-1}(e)\leq  \boldsymbol{opt}(e)$ and $\boldsymbol{x}_i(e)=\boldsymbol{s}_{i-1}(e)+\sigma_i > \boldsymbol{opt}(e)$. 
	
	Then we have 
	\begin{eqnarray}
	((\boldsymbol{opt} \lor \boldsymbol{x}_{i-1})\land \boldsymbol{y}_{i-1})-(\boldsymbol{opt} \lor (\boldsymbol{x}_{i-1}+\sigma_i \boldsymbol{\chi}_e)\land \boldsymbol{y}_{i})=-(\boldsymbol{s}_{i-1}(e)+\sigma_i-\boldsymbol{opt}(e))\boldsymbol{\chi}_e \nonumber
	\end{eqnarray}
	Let $\boldsymbol{s}_{i-1}(e)+\sigma_i -\boldsymbol{opt}(e)=\delta$,  then
	\begin{eqnarray}
	((\boldsymbol{opt} \lor \boldsymbol{x}_{i-1})\land \boldsymbol{y}_{i-1})-(\boldsymbol{opt} \lor (\boldsymbol{x}_{i-1}+\sigma_i \boldsymbol{\chi}_e)\land \boldsymbol{y}_{i})= -\delta  \boldsymbol{\chi}_e.\nonumber
	\end{eqnarray}
   And 
	\begin{eqnarray}
	\begin{aligned}
	&f((\boldsymbol{opt} \lor \boldsymbol{x}_{i-1}) \land \boldsymbol{y}_{i-1})-f((opt \lor (\boldsymbol{x}_{i-1}+\sigma_i  \boldsymbol{\chi}_e)) \land \boldsymbol{y}_{i})\\=&f(-\delta  \boldsymbol{\chi}_e|(\boldsymbol{opt} \lor (\boldsymbol{x}_{i-1}+\sigma_i  \chi_e))\land \boldsymbol{y}_{i})\\=&f(-\delta \boldsymbol{\chi}_e|\boldsymbol{opt}_i).
	\end{aligned}
	\end{eqnarray}
	 By the definition of the random variable $\boldsymbol{opt}_i$, we have
	\begin{eqnarray}
	\boldsymbol{x}_i\leq \boldsymbol{opt}_i\leq \boldsymbol{y}_i\nonumber
	\end{eqnarray}
Note that  $0\leq \delta < \sigma_i$. Since $\psi(b)  :=f(-\boldsymbol{\chi}_e|\boldsymbol{y}-b\boldsymbol{\chi}_e)$ is a non-increasing function on b. $ f(-\sigma_i \boldsymbol{\chi}_e|\boldsymbol{y}_i)\leq 0$ implies $f(-\delta \boldsymbol{\chi}_e|\boldsymbol{y}_i-\sigma_i \boldsymbol{\chi}_e)\leq 0$. We note that in the $i$the iteration,  $\boldsymbol{y}_i= \boldsymbol{y}_{i-1}$ and due to the condition of case 1.3, we have $\boldsymbol{y}_i-\sigma_i \boldsymbol{\chi}_e\geq \boldsymbol{opt}_i$, thus  $f(-\delta  \boldsymbol{\chi}_e|\boldsymbol{opt}_i)\leq f(-\delta \boldsymbol{\chi}_e|\boldsymbol{y}_i-\sigma_i \boldsymbol{\chi}_e)\leq 0\leq  \frac{\alpha_i}{2}$.
	
	\textbf{Case 2}: ($\alpha_i < 0 $ and $\beta_i>0$).This case is analogous to the previous one, and therefore we omit its proof.
	
	\textbf{Case 3}: ($\alpha_i \geq0 $ and $\beta_i>0$). With probability $\frac{\alpha_i}{\alpha_i+\beta_i}$ the following events happen: $\boldsymbol{x}_i \leftarrow \boldsymbol{x}_{i-1} + \sigma_i \boldsymbol{\chi}_e$ and $\boldsymbol{y}_i \leftarrow \boldsymbol{y}_{i-1}$; while with probability $\frac{\beta_i}{\alpha_i+\beta_i}$  the following events happen: $x_i \leftarrow \boldsymbol{x}_{i-1} $ and $\boldsymbol{y}_i\leftarrow \boldsymbol{y}_{i-1}-\sigma_i \boldsymbol{\chi}_e$ Thus,
	\begin{eqnarray}\label{eq2}
	\begin{aligned}
	\mathbb{E} [f(\boldsymbol{x}_i) -f(\boldsymbol{x}_{i-1})  +f(\boldsymbol{y}_i) -f(\boldsymbol{y}_{i-1})]&=\frac{\alpha_i}{\alpha_i+\beta_i}[f(\boldsymbol{x}_{i-1}+\sigma_i \boldsymbol{\chi}_e)-f(\boldsymbol{x}_{i-1})]\\&+\frac{\beta_i}{\alpha_i+\beta_i}[f(\boldsymbol{y}_{i-1}-\sigma_i \boldsymbol{\chi}_e)-f(\boldsymbol{y}_{i-1})]\\&=\frac{\alpha_i^2+\beta_i^2}{\alpha_i+\beta_i}
	\end{aligned}
	\end{eqnarray}
	Next, we upper bound $\mathbb{E}[f(\boldsymbol{opt}_{i-1})-f(\boldsymbol{opt}_{i})]$.
	
	\begin{eqnarray}\label{eq3}
	\begin{aligned}
	&\mathbb{E}[f(\boldsymbol{opt}_{i-1})-f(\boldsymbol{opt}_{i})]\\=&\frac{\alpha_i}{\alpha_i+\beta_i}[f((\boldsymbol{opt} \lor \boldsymbol{x}_{i-1}) \land \boldsymbol{y}_{i-1})-f((\boldsymbol{opt} \lor (\boldsymbol{x}_{i-1}+\sigma_i  \boldsymbol{\chi}_e)) \land \boldsymbol{y}_{i})]\\&+\frac{\beta_i}{\alpha_i+\beta_i}[f((\boldsymbol{opt} \lor \boldsymbol{x}_{i-1} ) \land \boldsymbol{y}_{i-1})-f((\boldsymbol{opt} \lor \boldsymbol{x}_{i}) \land (\boldsymbol{y}_{i-1}-\sigma_i \boldsymbol{\chi}_e))]\\\leq& \frac{\alpha_i\beta_i}{\alpha_i+\beta_i} 
	\end{aligned}
	\end{eqnarray}
	The final inequality follows by considering two cases. The first case is: $\boldsymbol{y}_i(e)=\boldsymbol{y}_{i-1}(e)-\sigma_i $, $ \boldsymbol{x}_i(e)=\boldsymbol{x}_{i-1}(e)$, the first term of the left-hand side of the last inequality equals zero. There are three subcases,
	
	\textbf{Case 3.1}:
	$(\boldsymbol{x}_i(e)\leq\boldsymbol{y}_i(e)=\boldsymbol{y}_{i-1}(e)-\sigma_i <\boldsymbol{y}_{i-1}(e)\leq \boldsymbol{opt}(e))$. 
	
	Thus
	$(\boldsymbol{opt} \lor \boldsymbol{x}_{i-1} ) \land \boldsymbol{y}_{i-1} = (\boldsymbol{opt} \lor \boldsymbol{x}_{i}) \land (\boldsymbol{y}_{i-1}-\sigma_i  \boldsymbol{\chi}_e)+\sigma_i\boldsymbol{\chi}_e $, and $(\boldsymbol{opt} \lor \boldsymbol{x}_{i}) \land (\boldsymbol{y}_{i-1}-\sigma_i  \boldsymbol{\chi}_e)\geq \boldsymbol{x}_{i-1}$, hence
	\begin{eqnarray}
	\begin{aligned}
	&f((\boldsymbol{opt} \lor \boldsymbol{x}_{i-1} ) \land \boldsymbol{y}_{i-1})-f((\boldsymbol{opt} \lor \boldsymbol{x}_{i}) \land (\boldsymbol{y}_{i-1}-\sigma_i  \boldsymbol{\chi}_e)) \\=& f(\sigma_i  \boldsymbol{\chi}_e|(\boldsymbol{opt} \lor \boldsymbol{x}_{i}) \land (\boldsymbol{y}_{i-1}-\sigma_i \boldsymbol{\chi}_e))\leq f(\sigma_i  \boldsymbol{\chi}_e|\boldsymbol{x}_{i-1})=\alpha_i\nonumber
	\end{aligned}
	\end{eqnarray}
	
	\textbf{Case 3.2}:
	$ (\boldsymbol{opt}(e)\leq \boldsymbol{y}_i(e)=\boldsymbol{y}_{i-1}(e)-\sigma_i <\boldsymbol{y}_{i-1}(e))$. 
	
	In this case, the second term of the left-hand side of inequality (\ref{eq3})  also equals zero, thus inequality (\ref{eq3})  follows.
	
	\textbf{Case 3.3}:
	$(\boldsymbol{x}_i(e)\leq\boldsymbol{y}_i(e)=\boldsymbol{y}_{i-1}(e)-\sigma_i \leq \boldsymbol{opt}(e)<\boldsymbol{y}_{i-1}(e))$. 
	
	We have
	$	(\boldsymbol{opt} \lor \boldsymbol{x}_{i-1} ) \land \boldsymbol{y}_{i-1} = (\boldsymbol{opt} \lor \boldsymbol{x}_{i}) \land (\boldsymbol{y}_{i-1}-\sigma_i  \boldsymbol{\chi}_e)+(\boldsymbol{opt}(e)-\boldsymbol{y}_{i-1}(e)+\sigma_i)\boldsymbol{\chi}_e $. Let $\mu = \boldsymbol{opt}(e)-\boldsymbol{y}_{i-1}(e)+\sigma_i$, then $0<\mu\leq \sigma_i$, and
	\begin{eqnarray}
	\begin{aligned}
	&f((\boldsymbol{opt} \lor \boldsymbol{x}_{i-1} ) \land \boldsymbol{y}_{i-1})-f((\boldsymbol{opt} \lor \boldsymbol{x}_{i}) \land (\boldsymbol{y}_{i-1}-\sigma_i  \chi_e)) \\=& f(\mu\boldsymbol{\chi}_e|(\boldsymbol{opt} \lor \boldsymbol{x}_{i}) \land (\boldsymbol{y}_{i-1}-\sigma_i  \boldsymbol{\chi}_e))\\\leq& f(\mu\boldsymbol{\chi}_e| \boldsymbol{y}_{i-1}-\sigma_i  \boldsymbol{\chi}_e)
	\leq 0\leq\alpha_i \nonumber
	\end{aligned}
	\end{eqnarray}
	The line 18-23 in the algorithm guarantees that  $f( -\boldsymbol{\chi}_e|\boldsymbol{y}_{i-1}-(\sigma -1)\boldsymbol{\chi}_e))\geq 0$. By lemma \ref{lemma3}, we have $f( -\sigma\boldsymbol{\chi}_e|\boldsymbol{y}_{i-1})\geq 0$.  Also $ f( -\mu\boldsymbol{\chi}_e|\boldsymbol{y}_{i-1}- (\sigma-\mu)\boldsymbol{\chi}_e)\geq 0$. Thus, 
		\begin{eqnarray}
	\begin{aligned}
  f( -\mu\boldsymbol{\chi}_e|\boldsymbol{y}_{i-1}- (\sigma-\mu)\boldsymbol{\chi}_e)= -f(\mu\boldsymbol{\chi}_e| \boldsymbol{y}_{i-1}-\sigma_i  \boldsymbol{\chi}_e) \geq 0\nonumber
		\end{aligned}
	\end{eqnarray}
By now, we show that the inequality (\ref{eq3}) holds for the case $\boldsymbol{y}_i(e)=\boldsymbol{y}_{i-1}(e)-\sigma_i $,  $ \boldsymbol{x}_i(e)=\boldsymbol{x}_{i-1}(e)$. The other case is that $\boldsymbol{y}_i(e)=\boldsymbol{y}_{i-1}(e) $,  $ \boldsymbol{x}_i(e)=\boldsymbol{x}_{i-1}(e)+\sigma_i$, which is analogous to the previous case, we omit the proof here.

	Now we show inequality (\ref{eq3})  follows for all situations. By (\ref{eq2}) and (\ref{eq3}) inequality (\ref{eq1}) holds if
	\begin{eqnarray}
	\frac{\alpha_i\beta_i}{\alpha_i+\beta_i} \leq \frac{1}{2} \cdot \frac{\alpha_i^2+\beta_i^2}{\alpha_i+\beta_i}\nonumber
	\end{eqnarray}
	which can easily be verified. 
\end{proof}	

\begin{theorem}
	Algorithm 1 is  a randomized $O(n\log B)$ time (1/2)-approximation algorithm for the DR-Submodular Maximization problem.
\end{theorem}

\begin{proof}
	Summing up lemma \ref{lemma4} for every $1 \leq i \leq \theta_e$ for each $e\in E$ gives
	\begin{eqnarray}
	\begin{aligned}\
	&\sum_{e \in E}\sum_{i=1}^ {\theta_e}\mathbb{E}[f(\boldsymbol{opt}_{i-1}^e)-f(\boldsymbol{opt}_{i}^e)]\\\leq& \frac{1}{2}\sum_{e \in E}\sum_{i=1}^ {\theta_e} \mathbb{E} [f(\boldsymbol{x}_i^e) -f(\boldsymbol{x}_{i-1}^e)  +f(\boldsymbol{y}_i^e) -f(\boldsymbol{y}_{i-1}^e)]
	\end{aligned}
	\end{eqnarray}
	The above sum is telescopic. We define that the algorithm executes on the vector coordinates ordered by $e_1,\dots,e_n$. Collapse the inequality, we get
	\begin{eqnarray}
	\begin{aligned}
	f(\boldsymbol{opt}_{0}^{e_1})-f(\boldsymbol{opt}_{\theta_e}^{e_n}) &\leq \frac{1}{2}\mathbb{E} [f(\boldsymbol{x}_{\theta_{e_n}}^{e_n}) -f(\boldsymbol{x}_{0}^{e_1})  +f(\boldsymbol{y}_{\theta_{e_n}}^{e_n}) -f(\boldsymbol{y}_{0}^{e_1})]\\&\leq \frac{1}{2}\mathbb{E}[f(\boldsymbol{x}_{\theta_{e_n}}^{e_n})+f(\boldsymbol{y}_{\theta_{e_n}}^{e_n})]
	\end{aligned}
	\end{eqnarray}
	Recalling the definitions of $\boldsymbol{opt}_i^e$,  $\boldsymbol{opt}_{0}^{e_1}=\boldsymbol{opt}$, $\boldsymbol{opt}_{\theta_e}^{e_n}=\boldsymbol{x}_{\theta_{e_n}}^{e_n} =\boldsymbol{y}_{\theta_{e_n}}^{e_n}$ is the output solution, thus $\mathbb{E}[f(\boldsymbol{x}_{\theta_{e_n}}^{e_n})]=\mathbb{E}[f(\boldsymbol{y}_{\theta_{e_n}}^{e_n})] \geq f(\boldsymbol{opt})/2$.
	It is clear that the algorithm makes  $O(n\log B)$ oracle calls since for each coordinate $e\in E$ the number of oracle calls is at most $\log B$ and there are $n=|E|$ coordinates.
\end{proof}

\section{Conclusions}
In this paper, we study the non-monotone DR-submodular maximization problem over bounded integer lattice, and specifically we study a Profit Maximization problem in social networks and show its DR-submodularity. By theoretically mining the special property of DR-submodular function,  we propose a binary search double greedy algorithm to the optimization problem. We show algorithm improves the approximation ratio and significantly reduces the time complexity compared to the existing results.



%
%
%
%

\end{document}